\documentclass[a4paper,UKenglish,cleveref, autoref, thm-restate]{lipics-v2021}
\usepackage{mathtools}
\usepackage{todonotes}
\usepackage{graphicx}
\usepackage{subcaption}

\newcommand{\Img}[0]{\mathrm{Im}}
\newcommand{\Lcal}[0]{\mathcal{L}}
\newcommand{\N}[0]{\mathbb{N}}

\newcommand{\Out}[0]{\mathrm{Out}}

\newcommand{\Cl}[0]{\mathit{Cl}}

\newcommand{\emptyword}{\epsilon}

\newcounter{problems}
\newtheorem{problem}[problems]{Problem}

\title{Minimization and Synthesis of the Tail in
Sequential Compositions of Mealy machines} 

\author{
Alberto Larrauri
}{
Graz University of Technology
}{}{}{}

\author{Roderick Bloem}{Graz University of Technology}{}{}{}

\authorrunning{
} 

\Copyright{Alberto Larrauri and Roderick Bloem} 



\category{} 

\relatedversion{} 

\acknowledgements{
}

\nolinenumbers 

\hideLIPIcs  

\EventEditors{John Q. Open and Joan R. Access}
\EventNoEds{2}
\EventLongTitle{32nd International Conference on Concurrency Theory}
\EventShortTitle{CVIT 2016}
\EventAcronym{CVIT}
\EventYear{2016}
\EventDate{December 24--27, 2016}
\EventLocation{Little Whinging, United Kingdom}
\EventLogo{}
\SeriesVolume{42}
\ArticleNo{23}

\begin{document}
\maketitle
\begin{abstract}

We consider a system consisting of a sequential composition of Mealy machines, called head and tail. We study two problems related to these systems. In the first problem, models of both head and tail components are available, and the aim is to obtain a replacement for the tail with the minimum number of states. We introduce a minimization method for this context which yields an exponential improvement over the state of the art.  In the second problem, only the head is known, and a desired model for the whole
system is given. The objective is to construct a tail that causes the system to behave according to
the given model. We show that, while it
is possible to decide in polynomial time whether such a tail exists, there are instances where its size is exponential in the sizes of the head and the desired system. This shows that the
complexity of the synthesis procedure is at least exponential, matching the upper bound in complexity
provided by the existing methods for solving unknown component equations.
\end{abstract}
\newpage
\section{Introduction}
\label{sec:intro}

The optimization of completely specified and incompletely specified Finite
State Machines are classical problems 
\cite{DBLP:books/aw/HopcroftU79,DBLP:journals/jacm/Ginsburg59,DBLP:journals/tc/PaullU59}. In practice, systems are usually decomposed into multiple FSMs that can be considered more or less separately. These decompositions lead to problems that have been studied intensively \cite{DBLP:books/daglib/0086041,harrisSynthesisFiniteState1998,villaUnknownComponentProblem2012}. First, to optimize a system, we can optimize its components separately. However, there may be more room for optimization
when considering them together. This way, it is possible to find $T$ and $T^\prime$ of different sizes that are not equivalent in isolation, but that can be used interchangeably in the context of another component $H$ \cite{unger1965flow}. 
Second, for a context $H$ and an overall desired behaviour $M$,  we may look for a component $T$ such that $T \circ H \equiv M$ \cite{benedettoModelMatchingFinitestate2001, watanabeMaximumSetPermissible1993a,villaUnknownComponentProblem2012}. This problem may occur, for instance, in \emph{rectification}, where a designer wants to repair or change the behaviour of a system by only modifying a part of it.

In this paper, we consider machines $H$ (head) and $T$ (tail), such that $T$ receives inputs from $H$ but not vice-versa. We study both the problem of minimizing $T$ without changing the language of the composition $T\circ H$ (Tail Minimization Problem) and the problem of building $T$ when we are given $H$ and a desired model for the system $M$
(Tail Synthesis Problem). We note that despite the simplicity of this setting, in various cases it is possible to tackle problems in more complex two-component networks using one-way compositions via simple reductions \cite{Wang_Brayton_1993}. We show that,
through polynomial reductions, our findings can be applied to any kind of composition between $H$ and $T$ where
all of $T$'s outputs can be externally observed. \par

The first exact solution to the\textbf{ Tail Minimization Problem} was given by Kim and Newborn
\cite{joonkikimSimplificationSequentialMachines1972}. They compute the smallest $T^\prime$ by minimizing
an incompletely specified (IS) Mealy machine $N$ with size proportional to $2^{|H|}|T|$. 
Minimization of IS machines is a known computationally hard problem \cite{Pfleeger_1973}.
Known algorithms for this task take exponential time in the size of $N$ \cite{Rho_Somenzi_Jacoby,
abelMEMINSATbasedExact2015,
penaNewAlgorithmExact1999},
which cannot be improved under
standard complexity-theoretic assumptions. This yields a doubly exponential complexity for the whole
Kim-Newborn procedure. This is unsatisfactory as deciding whether a
replacement $T^\prime$ for the tail $T$ of a given size exists is an NP problem: given a candidate $T^\prime$, we can compute the composition $T^\prime\circ H$ and check language equivalence with the original system $T\circ H$ in polynomial time. The reason for the blow-up in the complexity
of the Kim-Newborn method lies in a determinization step that results in $N$ being exponentially larger
than $H$. Various non-exact methods have been studied to avoid this determinization step. For instance, Rho and Somenzi \cite{june-kyungrhoDonCareSequences1994} present an heuristic in which a ``summarized''
incompletely specified machine is obtained, and in \cite{huey-yihwangMultilevelLogicOptimization1995}
Wang and Brayton avoid performing exact state minimization and in turn perform optimizations at the net-list logic level. \par

The observation that the Tail Minimization Problem can be solved efficiently with access to an NP oracle yields a straightforward solution via iterative encodings into SAT.
As an improvement, we propose a
modification of the Kim-Newborn procedure and that runs in singly exponential time. To develop our method, we introduce the notion of observation machines (OMs), which can be regarded as IS machines with universal branching. Our algorithm yields a polynomial reduction of the Tail Minimization Problem to the problem
of minimizing an OM.
To carry out this last task, we generalize the
procedure for minimizing IS machines shown in \cite{abelMEMINSATbasedExact2015} to OMs. Preliminary experimental results show that the proposed approach is much more efficient than the naive encoding into SAT on random benchmarks.
\par

We then turn to the \textbf{Tail Synthesis Problem}. Here, we are given $H$ and $M$ and we look for a $T$ such that $T \circ H \equiv M$.
Apart from finding such $T$, we may simply be interested in deciding whether it exists. In this case we say that the instance of the problem is feasible.
This problem is a particular case of
so-called missing component equations \cite{villaUnknownComponentProblem2012},
where arbitrary connections between $H$ and $T$
are allowed. The general approach given in \cite{watanabeMaximumSetPermissible1993a},
is based on the construction of
a deterministic automaton representing 
the ``flexibility'' of $T$, called the E-machine. 
Both the task of deciding whether an equation is feasible and the one of computing a solution in the affirmative case take linear time in the size of the E-machine. However, this automaton
is given by a subset construction, and has size $O(2^{|H|}|M|)$. This yields an
exponential upper bound for the complexity of solving general unknown component equations as well as the complexity of only deciding their feasibility. \par
We show the surprising result that deciding the feasibility of the Tail Synthesis Problem takes polynomial time, but computing an
actual solution $T$ has exponential complexity, matching the bound given by the E-machine approach
\cite{watanabeMaximumSetPermissible1993a}. This follows from the existence of
instances of the synthesis problem where all solutions $T$ have exponential size. We give
a family of such instances constructively. Additionally, we show that it is possible
to represent all solutions of the Tail Synthesis Problem via an OM. This representation avoids any kind of subset construction and hence is exponentially more succinct than the E-machine. \par



\section{Preliminaries}

\paragraph*{General Notation}
We write $[k]$ for $\{0,\dots,k-1\}$, $2^X$ for the power set of $X$, and $X^*$
for the set of finite words of arbitrary length over $X$.
We use overlined variables $\overline{x} = x_0 \dots x_{n-1}$ for words, and write
$\emptyword$ for the empty word.
Given two words 
of the same length 
we define $\langle \overline{x}, \overline{y} \rangle
\coloneqq (x_0,y_0)\dots (x_n, y_n) \in (X\times Y)^*$.

\paragraph*{Mealy Machines}
\label{sec:mealy}
Let $X$ and $Y$ be finite alphabets.
A \textbf{Mealy machine} $M$ from $X$ to $Y$ is a tuple 
$(X, Y, S_M, D_M,$ $\delta_M,\lambda_M, r_M)$, where $S_M$ 
is a finite set of states,
$D_M\subseteq S_M\times X$ is a specification domain,
$\delta_M: D_M \rightarrow S_M$ 
is the next state function, $\lambda_M: D_M \rightarrow Y$ 
is the output function and $r_M\in S_M$ is the initial 
state. We say that an input string $\overline{x}\coloneqq x_0x_1\dots
x_n \in X^*$ is \textbf{defined} at a state
$s_0$ if there are states 
$s_1,\dots, s_{n+1}$ 
satisfying both $(s_i,x_i)\in D_M$ and $s_{i+1}=\delta_M(s_i,x_i)$ 
for all $0\leq i \leq n$. We call the sequence 
$s_0,x_0,y_0,s_1, \dots,  s_n,x_n,y_n,s_{n+1}$,
where $y_i=\lambda_M(s_i,x_i)$,
the \textbf{run} of $M$ on $\overline{x}$
from $s_0$. When $s_0=r_M$, we simply
call this sequence the run of $M$ on $\overline{x}$. 
We write $\Omega_M(s)$ for the set of input sequences
defined at state $s$, and $\Omega_M$ for $\Omega_M(r)$.
We lift $\delta_M$ and $\lambda_M$
to defined input sequences in the natural way. 
We define $\delta_M(s, \emptyword)=s$,
$\lambda_M(s,\emptyword)=\emptyword$ for all $s$. 
Given $\overline{x} \in \Omega_M(s)$, if 
$s^\prime = \delta_M(s, \overline{x})$
and  $(s^\prime,x^\prime)\in D_M$ then
$\delta_M(s,\overline{x} x^\prime)=\delta_M(s^\prime,x^\prime)$ and
$\lambda_M(s,\overline{x}x^\prime) = 
\lambda_M(s,\overline{x})\lambda_M(s^\prime,x^\prime)$. We write $\delta_M(\overline{x})$ and $\lambda_M(\overline{x})$
for $\delta_M(r_M,\overline{x})$ and $\lambda_M(r_M,\overline{x})$ 
respectively. We define $\Out(M)$
as the set of words $\lambda_M(\overline{x})$, for
all $\overline{x}\in \Omega_M$.\par

We say $M$ is \textbf{completely specified}, if $D_M=S_M\times X$.
Otherwise we say that $M$ is \textbf{incompletely specified}. From
now on, we refer to completely specified Mealy machines simply as
Mealy machines, and to incompletely
specified ones as IS Mealy machines. 
For considerations of computational complexity, we consider alphabets to be fixed and the size $|M|$ of a Mealy machine $M$ to be proportional to its number of states. \par

We say that a (completely specified) Mealy machine $N$
\textbf{implements} an IS Mealy machine $M$ with the same input/output
alphabets as $N$ if
$\lambda_N(\overline{x})=
\lambda_M(\overline{x})$
for all 
$\overline{x}\in \Omega_M$.
The problem of minimizing an IS 
Mealy machine $M$ consists of 
finding a minimal implementation for it. It is well known that this
problem is computationally hard, in contrast with the minimization of completely specified Mealy machines. More precisely, we say that $M$ is \textbf{$n$-implementable} if it has some implementation whose size is not greater than $n$. Given an IS machine $M$
and some $n$, deciding whether 
$M$ is $n$-implementable is an NP-complete problem \cite{Pfleeger_1973}.

\paragraph*{Automata Over Finite Words}
\label{sec:safety_automata} 
We consider automata over finite words
where all states 
are accepting. Let $\Sigma$ be a finite alphabet. A \textbf{non-deterministic finite automaton (NFA)} $A$ \textbf{over} $\Sigma$,
is a tuple $(\Sigma, S_A, \Delta_A, r_A)$, 
where $S_A$ is a finite set of states, $\Delta_A: S_A\times \Sigma \rightarrow 2^{S_A}$ 
is the transition function, and $r_A\in S_A$ is the 
initial state.  A \textbf{run} of $A$ on a word $\overline{a}$ is defined as usual. We say that $A$ \textbf{accepts} a word $\overline{a}$
if there is a run of $A$ on $\overline{a}$. 
The \textbf{language} of 
$A$ is the set $\Lcal(A)\subseteq \Sigma^*$ containing the words accepted by $A$. Note that $\Lcal(A)$
is prefix-closed. We lift $\Delta_A$
to words $\overline{a}\in \Sigma^*$ and sets of states $Q\subseteq S_A$
in the natural way. The set
$\Delta_A(s,\overline{a})$ consists
of all $s^\prime$ such that some run of $A$ on $\overline{a}$
from $s$ finishes at $s^\prime$,
and we write $\Delta_A(Q,\overline{a})$ for $\cup_{s\in Q} \Delta_A(s,\overline{a})$. We write $\Delta_A(\overline{a})$ for $\Delta_A(r_A,\overline{a})$.
We say that $A$ is \textbf{deterministic (a DFA)}, if 
$|\Delta(s,a)|\leq 1$ for all $s,a$.
For considerations of computational complexity, we consider the alphabet to be fixed, as before, and the size $|A|$ of an automaton $A$ to be proportional to $|S_A| + \sum_{s,a} |\Delta_A(s,a)|$.

\section{Problem statements}

We define a \textbf{cascade composition of Mealy machines} 
as a system
consisting of two Mealy machines $H$,
the \textbf{head},
and $T$,
the \textbf{tail}, that work in sequential composition as shown in \cref{fig:cascade}). We write $T\circ H$ to refer to this sequential composition. The behaviour such cascade composition can be described via another Mealy machine $M$ resulting from a standard product construction: Set $S_M\coloneqq S_H \times S_T$ and $r_M\coloneqq (r_H,r_T)$. Let 
$s_H\in S_H, s_T\in S_T, x\in X$
and $y\coloneqq \lambda_H(s_H, x)$. We define $\delta_M( (s_H, s_T), x)\coloneqq (s^\prime_H, s^\prime_T)$ where $s^\prime_H\coloneqq\delta_H(s_H,x)$ and
$s^\prime_T\coloneqq 
\delta_T(s_T,y)$, and 
$\lambda_M( (s_H, s_T), x)\coloneqq \lambda_T(s_T,y)$. Given a cascade composition 
$T\circ H$, we say that a Mealy machine $T^\prime$ is a 
\textbf{replacement} for $T$ if $T^\prime \circ H\equiv T\circ H$.
We say that $T$ is \textbf{$n$-replaceable}
in $T\circ H$ if there is a replacement
$T^\prime$ for $T$ with at most $n$
states. 

\begin{figure}[tb]
    \centering
    \begin{subfigure}{0.40\textwidth}
    \includegraphics[width=\textwidth]{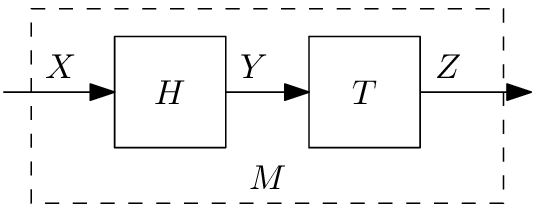}
    \caption{A cascade composition of Mealy machines.}
    \label{fig:cascade}
    \end{subfigure}
    \hfill
    \begin{subfigure}{0.40\textwidth}
    \includegraphics[width=\textwidth]{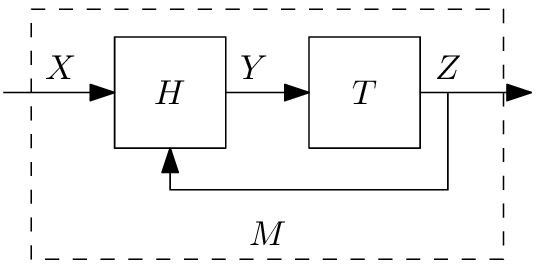}
    \caption{A composition between Mealy machines with two-way communication.}
    \label{fig:general_cascade}
    \end{subfigure}
    \caption{Two different compositions of Mealy machines.}
    
\end{figure} 

We study two problems related to the tail component $T$ of a cascade composition. In the first, both
$H$ and $T$ are given and the goal is to find a minimal replacement $T^\prime$ for $T$ that leaves 
the behaviour of the system unaltered.
In the second, $H$ and $M$ are given instead and one is asked to find $T$ such that $T\circ H\equiv M$. We also consider two related decision problems.

\begin{problem}[Tail Minimization Problem]
\label{prob:minim}
Given a cascade composition $T\circ H$, find a replacement $T^\prime$ for 
$T$ with the minimum amount of states.
\end{problem}

\begin{problem}[Tail Synthesis Problem]
\label{prob:equation}
Given Mealy machines $H$ and $M$ sharing the same input alphabet, construct a Mealy machine $T$ so that $T\circ H\equiv M$.
\end{problem}

\begin{problem}[$n$-Replaceabilty of the Tail]
\label{prob:n_replaceability}
Given a cascade composition $T\circ H$ and a number $n\in \N$, 
decide whether there is a replacement $T^\prime$ for $T$ with at most
$n$ states.
\end{problem}

\begin{problem}[Feasibility of the Tail Synthesis Problem]
\label{prob:feasability}
Given Mealy machines $H$ and $M$ with the same input alphabet, decide whether there exists some Mealy
machine $T$ such that $T\circ H\equiv M$.
\end{problem}

\section{Observation Machines}
\label{sec:ofa_prelim}
We will now define observation machines, which
can be regarded as IS Mealy machines 
\cite{harrisSynthesisFiniteState1998} with 
universal branching. 
This construction allows us to express the 
solutions of the Tail Minimization Problem
and the Tail Synthesis 
Problem 
while avoiding the 
determinization
steps in \cite{joonkikimSimplificationSequentialMachines1972} and
\cite{watanabeMaximumSetPermissible1993a}.
An \textbf{observation machine (OM)} from $X$ to $Y$
is a tuple $M=(X,Y,S_M,D_M,$ 
$\Delta_M,\lambda_M,r_M)$
defined the same way as a Mealy machine except 
for the 
next-state function $\Delta_M$, which now maps
elements of the specification domain
to sets of states $\Delta_M: D_M \rightarrow 2^{S_M} 
\setminus \{\emptyset\}$. 
A \textbf{run of $M$ over $\overline{x}\in X^*$ starting from $s_0\in S_M$},
is a sequence
$s_0,x_0,y_0,s_1,\dots,s_n,x_n,y_n,s_{n+1}$, where
$(s_i,x_i)\in D_M$, 
$s_{i+1}\in \Delta_M(s_i,x_i)$ and
$y_i=\lambda(s_i,x_i)$ for all $0\leq i\leq n$. We call such run simply
a run of $M$ on $\overline{x}$
when $s_0\coloneqq r_M$. We 
say that a sequence $\overline{x}$
is defined at a state $s$ if there is a run of $M$ on $\overline{x}$
starting from $s$. As with Mealy machines, we put $\Omega_M(s)$ and $\Omega_M$ for the sets
of defined sequences at $s$ and at $r_M$ respectively. 
We say that $M$ is \textbf{consistent} if all runs of $M$ on 
any given defined input sequence $\overline{x}\in \Omega_M$ have the same output. 
Unless otherwise specified, we assume all OMs to be consistent. 
In this case we can lift $\lambda$ to defined
input sequences, witting $\lambda_M(\overline{x})$
for the unique output sequence $\overline{y}$ corresponding to all runs of $M$ over $\overline{x}$.
We also lift the transition
function $\Delta_M$ to defined sequences and sets of states as done previously for NFAs.
Note that when $M$ has no branching, i.e. $|\Delta_M(s,x)|=1$ for all 
$(s,x)\in D_M$, we end up with a construction equivalent to an IS Mealy machine. 
Again, for considerations of computational complexity, we consider alphabets to be fixed, and the size $|M|$ to be proportional to $|S_M| + \sum_{(s,a)\in D_M} |\Delta_M(s,a)|$.
\par
We can use consistent OMs to represent specifications over  Mealy machines. 
We say that a machine $N$ \textbf{implements} a OM $M$ with the
same input/output alphabets as $N$
if $\lambda_N(\overline{x})=\lambda_M(\overline{x})$ for all $\overline{x}\in
\Omega_M$. Note that it is straightforward to test this property in $O(|N||M|)$ time 
via a product construction. 
We say that $M$
is $n$-implementable if it has an
implementation with at most $n$ states. The problem of \textbf{minimizing} an OM $M$
is the one of finding an implementation of it with the minimum number of states. 
\par
Informally, the reason we call the branching ``universal'' is that the i/o language of an OM $M$ is given by an universal automata of the same size, but not by a non-deterministic one. The intuitive argument is that once an input word leaves $\Omega_M$,
all behaviours are allowed. In a way, this means that $M$ accepts the complement of $\Omega_M$, which is given by an NFA.

\section{Optimization of the Tail}

In this section we introduce a novel
solution for the Tail Minimization
Problem. This solution improves over
the state of the art, 
represented by the Kim-Newborn (K-N) method \cite{joonkikimSimplificationSequentialMachines1972}, by avoiding an expensive
determinization step during the process. Two important observations
are the following.

\begin{observation}
\label{obs:npmin}
The $n$-replaceability decision
problem \cref{prob:n_replaceability} is in  NP: Given a candidate $T^\prime$ with $|S_{T^\prime}|\leq n$ it takes polynomial 
time to build Mealy models 
for $T\circ H$ and $T\circ H^\prime$,  and to decide whether they are equivalent.
\end{observation}

\begin{observation}
\label{lemm:main}
A Mealy machine $T^\prime$
is a replacement for $T$ if and only if 
$\lambda_T(\overline{y})=
\lambda_{T^\prime}(\overline{y})$
for all $\overline{y}\in \mathrm{Out}(H)$
\end{observation}

Because of  \cref{obs:npmin},
there are
exponential-time algorithms for \cref{prob:minim} and
\cref{prob:n_replaceability}:
there is a straightforward (``naive'') polynomial reduction of the $n$-replaceability problem into a satisfiability problem along the lines of bounded synthesis \cite{DBLP:journals/sttt/FinkbeinerS13}. This encoding can be used to optimize the tail of a cascade composition by
finding the minimum $n$ for which the resulting CNF formula is satisfiable. The exponential complexity of this procedure contrasts with the double-exponential complexity of
the K-N algorithm. This reasoning also applies to more general
networks of Mealy machines 
\cite[Chapter 6]{harrisSynthesisFiniteState1998}, implying that the approach for component minimization
based on optimizing the ``$E$-machine'' 
\cite{harrisSynthesisFiniteState1998, watanabeMaximumSetPermissible1993a}
is not optimal in theory, as it takes doubly exponential time. \par
In the following, we study the problem in more detail and present an novel approach that we will compare against the naive encoding in the experimental section.

\begin{figure}[tb]
\vspace{0.2cm}
    \centering
\begin{subfigure}{0.40\textwidth}
    \includegraphics[width=\textwidth]{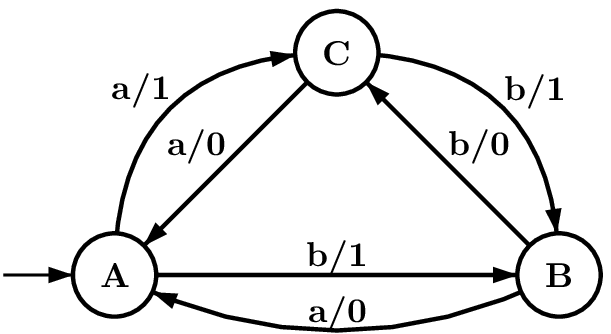}
    \vspace{0.2cm}
    \caption{A Mealy machine $H$ from
    $X\coloneqq \{a,b\}$ to $Y\coloneqq 
    \{0,1\}$.}
    \label{fig:h}
\end{subfigure}
\hfill
\begin{subfigure}{0.40\textwidth}
    \includegraphics[width=\textwidth]{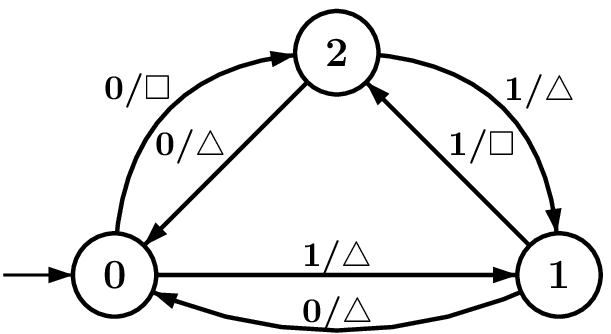}
        \vspace{0.2cm}
    \caption{A Mealy machine $T$ from
    $Y\coloneqq 
    \{0,1\}$ to $Z\coloneqq \{\triangle,\square\}$.
    }
    \label{fig:T}
\end{subfigure}

\caption{The head $H$ and the tail $T$ of a cascade composition.}
\label{fig:example}
\end{figure}

\subsection{Proposed Minimization Algorithm}
\label{sec:min_alg}

We give an overview of
our minimization method here. The algorithm is divided in three steps: 
(1) We compute an NFA $A$ which accepts the language $\mathrm{Out}(H)$ 
, (2)
using $A$ and $T$ we build an OM $M$ whose set of implementations is precisely the set of replacements for $T$ 
. By \cref{lemm:main},
this is ensured by
$\Omega_M=\Out(H)$
and $\lambda_M(\overline{y})=
\lambda_T(\overline{y})$ for
all $\overline{y}\in \Omega_M$.
Lastly, (3) we find an implementation $T^\prime$ of
$M$ with the minimum number of states 
. 
This machine $T^\prime$ is a minimal replacement for $T$. \par
This algorithm follows the K-N procedure but skips an expensive determinization step. Indeed, the difference is that in
the K-N method the NFA $A$ 
obtained in (1) is determinized 
before performing (2). This removes the branching from $M$,
making it an IS machine, but it yields $|M|=O(2^{|H|}|T|)$ rather
than $|M|=O(|H||T|)$ in our method. Step (3), the last one, is responsible of the overall complexity of our algorithm as well as the K-N one, and takes $2^{|M|^{O(1)}}$ time. This
makes the total time costs
of our procedure and the K-N one
$2^{(|H||T|)^{O(1)}}$ and
$2^{(2^{|H|}|T|)^{O(1)}}$, respectively. 
We illustrate our method in \cref{fig:example2} and \cref{fig:example3},
where we minimize the tail of the cascade
composition shown in \cref{fig:example}.

\paragraph*{The Image Automaton}
Let $H$ be a Mealy machine from $X$ to $Y$.
In this section we describe how to obtain 
an NFA whose language is  
$\Out(H)$ \cite{hennie1968finite,joonkikimSimplificationSequentialMachines1972}. The \textbf{image automaton} of $H$ (sometimes called the inverse automaton), written
$\Img(H)$, is the NFA over $Y$ defined as follows: Let
$S_{\Img(H)}\coloneqq S_H$ and $r_{\Img(H)}\coloneqq r_H$. 
We set 
$\Delta_{\Img(H)}(s_1,y)\coloneqq\{
s_2 \in S_H \, \mid \, \exists x\in X \, \text{s.t. } \lambda(s_1,x)=y,\, \delta(s_1,x)=s_2 \, \}$. 
It holds that
$\Lcal(\Img(H))=\mathrm{Out}(H)$.
Essentially, to obtain $\Img(H)$ one deletes 
the input labels from $H$'s transitions, as shown
in \cref{fig:inv}. The time and space
complexity of this construction is $O(|H|)$. \par


\begin{figure}[tb]
    \centering
    \begin{subfigure}{0.40\textwidth}
    \includegraphics[width=\textwidth]{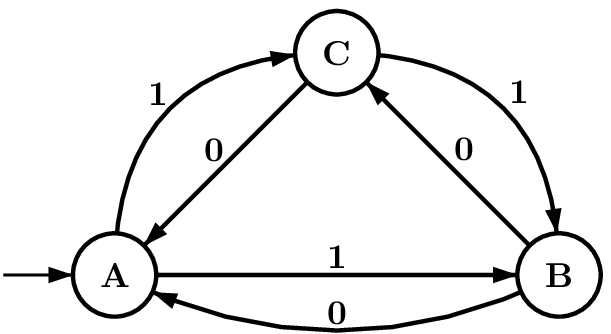}
    \vspace{0.3cm}
    \caption{The image automaton $\Img(H)$ of the machine $H$
    in \cref{fig:h}.}
    \label{fig:inv}
    \end{subfigure}
    \hfill
     \begin{subfigure}{0.45\textwidth}
    \includegraphics[width=\textwidth]{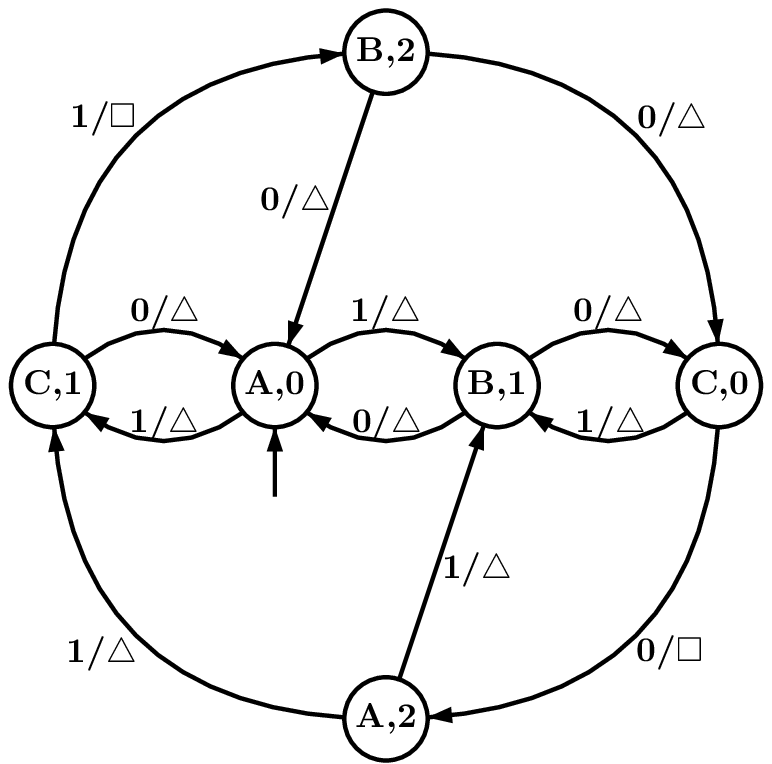}
    \vspace{0.3cm}
    \caption{The restriction $M\coloneqq 
    T|_{\Img(H)}$ of the Machine $T$ in \cref{fig:T}
    to the automata $A$ in \cref{fig:inv}.}
    \label{fig:flex}
    \end{subfigure}
    \caption{First and second step of the minimization of $T$
    in \cref{fig:example}}
    \label{fig:example2}
\end{figure}

\paragraph*{The Restriction Machine}
Let $A$ be an NFA over $Y$,
and $T$ a machine from $Y$ to $Z$.
In this section our goal is to build an
OM $M$ whose set of defined sequences
$\Omega_M$ is precisely $\Lcal(A)$ and
that satisfies $\lambda_M(\overline{y})=\lambda_T(\overline{y})$
for all $\overline{y}\in \Omega_M$.
The \textbf{restriction of } $T$
\textbf{ to } $A$, denoted by
$T|_A$, is the OM $M$ from $Y$ to $Z$
defined as follows. Let $S_M \coloneqq S_T\times S_A$ and $r_M\coloneqq (r_T,r_A)$. Given a state 
$s_M\coloneqq(s_T, s_A)\in 
S_M$ and an input $y\in Y$ there are two 
possibilities: (1) $\Delta_A(s_A,y)\neq \emptyset$. In
this case we mark the transition as defined
$(s_M, y)\in D_M$, and set 
$
\Delta_M(s_M, y)=
\{
(s^\prime_T,s^\prime_A) \, \mid \, 
s^\prime_T=\delta_T(s_T,y), s^\prime_A\in 
\Delta_A(s_A,y)
\}$.
Alternatively, (2) $\Delta_A(s_A,y) = \emptyset$.
Here we just mark the transition as undefined
$(s_M, y)\notin D_M$. 
It is direct to see 
that both $\Omega_M=\Lcal(A)$ and
$\lambda_M(\overline{y})=\lambda_T(\overline{y})$
for all $\overline{y}\in \Omega_M$. 
\par

An example is given in \cref{fig:flex}. This product construction generalizes the one in the K-N algorithm: when $A$ is deterministic,
so is $M$, yielding an IS Mealy machine. The construction of $M\coloneqq T|_A$ can be performed in $O(|T||A|)$ time. In the case where $A\coloneqq \Img(H)$ results from the head $H$ of a cascade, we can substitute $|A|=O(|H|)$.



\paragraph*{Reduction to a Covering Problem}
\label{sec:covering}
 
Let $M$ be an OM from $Y$ to $Z$. Our
objective is to find a minimal implementation of $M$. In order to do this, we generalize the theory
of \cite{grasselliMethodMinimizingNumber1965} for minimization of IS machines. The basic idea is that
we can define a compatibility relation
$\sim$ over $S_M$ and use it to reduce the task to a covering problem over $S_M$. Two states 
$s_1,s_2\in S_M$ are \textbf{compatible},
written $s_1\sim s_2$, if $\lambda_M(s_1,\overline{y})=\lambda_M(s_2,\overline{y})$ for all $\overline{y}\in \Omega_M(s_1)\cap
\Omega_M(s_2)$. We note that this
relation between states is symmetric and reflexive, but not necessarily transitive.
A set $Q\subseteq S_M$ is called
a \textbf{compatible} if all its states are pairwise
compatible. We say that $Q$ 
is \textbf{incompatible at  depth} $k$
if $k$ is the length of the shortest word $\overline{y}\in Y^*$ such that
$\overline{y}$ is defined for two states
$s_1,s_2\in Q$, and $\lambda_M(s_1,\overline{y})\neq
\lambda_M(s_2,\overline{y})$. A convenient characterization is as follows:

\begin{lemma}
\label{lem:compat_2}
Let $Q\subseteq S_M$. The following statements hold: 
(1) $Q$ is incompatible at depth $1$ if and only if for some $s_1,s_2\in Q$ there is a defined input
$y\in Y$
with $\lambda_M(s_1,y)\neq \lambda_M(s_2,y)$.
(2) If $Q$ is incompatible at depth $i>1$
then $\Delta_M(Q,y)$ is incompatible at depth $i-1$ for some $y\in Y$.
\end{lemma}

\Cref{lem:compat_2} gives 
a straightforward way to compute the compatibility relation $\sim$ over $S_M$. We begin by finding
all pairs $s_1,s_2$ incompatible at 
depth $1$. Afterwards we propagate
the incompatible pairs backwards: 
if $s_1\nsim s_2$ and $s_1\in \Delta_M(s_1^\prime,y), s_2\in \Delta_M(s_2^\prime,y)$, then $s_1^\prime\nsim s_2^\prime$. This process can be carried out in $O(|M|^2)$ time (see \cref{ann:graph} for a reduction to Horn SAT). When $M=T|_{\Img(H)}$, it holds
$|M|=O(|H||T|)$, and the
required time to compute
the $\sim$ relation over $B$ is $O(|H|^2|T|^2)$. \par

Let $F\subseteq 2^S_M$ be a family where all $C\in F$ are compatibles. 
We call $F$ a \textbf{closed cover of compatibles over} $M$ if the following
are satisfied: (1) $r_M\in C$ for some $C\in F$, and (2)
for any $C\in F$ and $y \in Y$ there is some $C^\prime \in F$ such that
$\Delta(C, y)\subseteq C^\prime$. The following theorem implies
that the problem of 
finding a minimal implementation for  $M$ is polynomially equivalent
to the problem of finding a closed cover of compatibles over it with the
minimum size. This equivalence is illustrated in \cref{fig:graph}.

\begin{theorem} 
\label{thm:cover_reduction}
Let $M$ be an OM from $Y$ to $Z$. 
Let $|M|\coloneqq \sum_{s,y} |\Delta_B(s,y)|$ 
be the number of transitions
in $M$. The following two statements hold:
(1) Let $N$ be an implementation of $M$. Then it
is possible to build a closed cover of compatibles $F$ over $M$ with 
$|F|\leq |S_N|$ in $O(|M||N|)$ time. (2) Let 
$F$ be a closed cover of compatibles over $M$. Then it is possible to build an 
implementation $N$ of $M$ with $|S_M|=|F|$ in $O(|M||F|)$ time.
\end{theorem}
\begin{proof} 
We prove (1) and (2) separately.
\textbf{(1): } Let $N$
be an implementation of $M$.
For each $s_N\in S_N$ we define the set $Q(s_N)\subseteq S_M$ as follows:
$Q(s_N)=\{ \, s_M\in S_M \, \mid \, \exists \overline{y}
\in \Omega_M \text{ s.t. } \delta_N(\overline{y})=s_N, \,
s_M \in \Delta_M(\overline{y})\, \}$.
It holds that each $Q(s_N)$ is a compatible, 
as the output of each state $s_M\in Q(s_N)$ over a defined
input sequence has to coincide with the output of $s_N$
in $N$.
Moreover, $r_M\in Q(r_N)$, and for any $s_N\in S_N$, $y\in Y$ it holds
that $\Delta_M(Q(s_N),y)\subseteq Q(\delta_M(s_N,y))$. Thus, the family 
$F\coloneqq\{Q(s_N) \, \mid \, s_N\in S_N \}$ is a closed cover of compatibles over $N$, and $|F|\leq |S_N|$.
Note that $|F|$ may be strictly lesser
than $|S_N|$, as for some $s^1_N, s^2_N\in S_M$ it may happen that 
$Q(s^1_M)=Q(s^2_M)$. The sets $Q(s_N)$, 
and with them the family $F$, can be 
derived from a product construction between $M$ and $N$
which takes $O(|M||N|)$ time.
\textbf{(2):} Let $F$ be a closed cover of compatibles over $M$. 
We build an implementation $N$ of $M$:
Set $S_N\coloneqq F$. Thus, each state in $N$ corresponds to compatible $C\in F$.
We define $r_N$ as an arbitrary $C\in F$
satisfying $r_M\in C$. 
Let $C\in F$ and $y\in Y$.
We define $\delta_N(C,y)$ as an arbitrary
$C^\prime \in F$ such that $\Delta_M(C,y)\subset C^\prime$. To define $\lambda_N(C,y)$ we
take into account two possibilities: (1) 
$(s_M,y)\in D_M$ for some $s_M\in C$. Then we
set $\lambda_N(C,y)=\lambda_M(s_M,y)$. 
Note that $\lambda_N(C,y)$ is independent of the 
particular choice of $s_M$, as $M$ being consistent
implies that all choices yield the same output. 
(2) $(s_M,y)\notin D_M$ for all $s_M\in C$. In this
case put an arbitrary output for $\lambda_N(C,y)$. 
It can be checked that $N$
is indeed an implementation of $M$ 
(\cref{ann:covering}). The Mealy machine $N$ given by this construction satisfies $|S_N|=|F|$,
and can be built in $O(|M||N|)$ time and space. 
\end{proof}

\begin{figure}[tbh]
    \centering
    \begin{subfigure}{0.40\textwidth}
    \includegraphics[width=\textwidth]{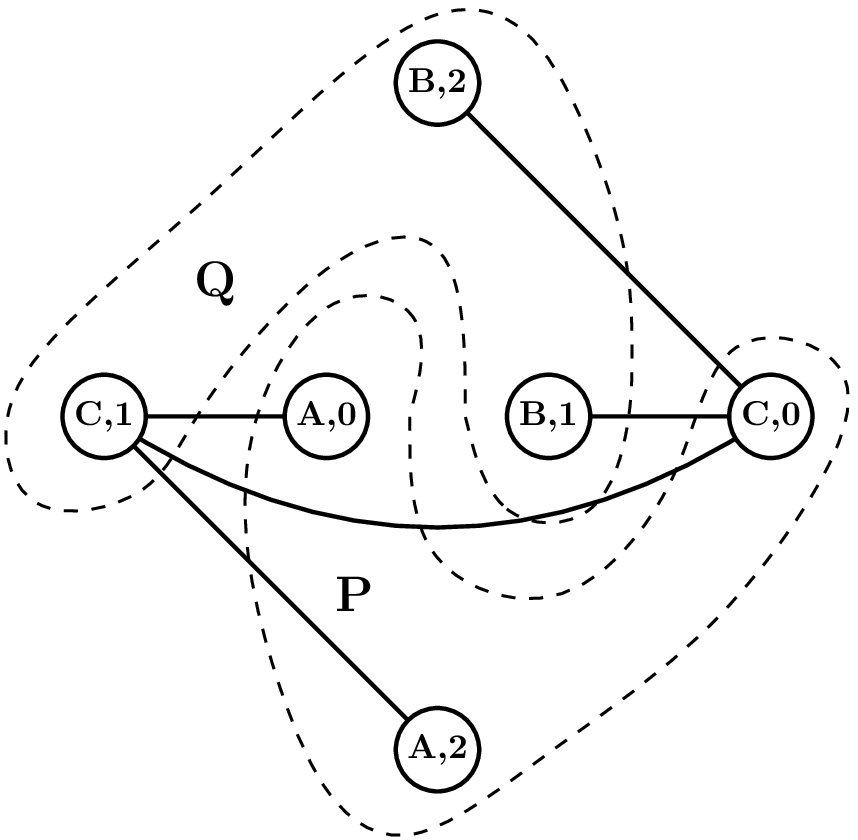}
    \vspace{0.3cm}
    \caption{The relation of the OM $M$
    in \cref{fig:flex}, together with a closed cover of 
    compatibles $P$ and $Q$}
    \label{fig:graph}
    \end{subfigure}
    \hfill
     \begin{subfigure}{0.50 \textwidth}
    \includegraphics[width=\textwidth]{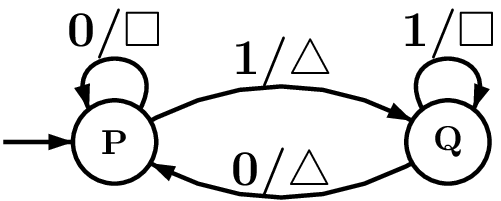}
    \vspace{0.3cm}
    \caption{A Mealy machine $T^\prime$ corresponding to
    the cover in \cref{fig:graph}.}
    \label{fig:mini}
    \end{subfigure}
    \caption{Final step of the minimization of $T$
    in \cref{fig:example}. The machine $T^\prime$ in
    \cref{fig:mini} is a minimal replacement for $T$.}
    \label{fig:example3}
\end{figure}



\paragraph*{Obtaining a Minimal Replacement}

Let $M$ be an OM from $Y$ to $Z$. Further 
suppose that $M=T|_{\Img(H)}$ for some machines $T,H$. Replacements of $T$ in
$T\circ H$ are precisely the implementations of $M$. Thus, solving the Tail Minimization Problem amounts to finding a minimal implementation of $M$. As shown in the previous section,
this is equivalent to finding a minimal closed
cover of compatibles over $M$. This reduction
has been widely employed in the study of the analogous minimization problem for IS Mealy machines \cite{Rho_Somenzi_Jacoby,
abelMEMINSATbasedExact2015,
Kam_Villa_Brayton_Sangiovanni-Vincentelli_1994}. 
We propose an adaptation of the method in \cite{abelMEMINSATbasedExact2015} as a tentative approach for finding a minimal implementation of $M$. \par
Given a bound $n$, we reduce the problem of
finding a closed cover $\mathcal{F}$ of $n$ compatibles over $M$ to a SAT instance. The CNF encoding follows closely the one in \cite{abelMEMINSATbasedExact2015},
and the details can be found at \cref{ann:minim_cnf}.
As them, we also compute in advance a so-called partial solution. This is clique of pairwise incompatible states $\mathit{Cl}\subseteq S_M$ obtained via a greedy algorithm. Each state
in $\mathit{Cl}$ must belong to a different
compatible in $\mathcal{F}$, so the clique can be used for adding symmetry breaking predicates to the encoding and thus to reduce solving times. \par
In order to obtain the minimal replacement for $T$ in $T\circ H$ we look for the minimum $n$
that yields a satisfiable CNF encoding. It is clear that such $n$ must lie in the interval $[|\Cl|,|S_T|]$. In particular, when $|\Cl|=|S_T|$, the machine $T$ is already optimal and no encoding is needed. Otherwise, any searching strategy
for $n$ in $[|\Cl|,|S_T|]$ may be employed.
In our case, we simply employ linear search from $|\Cl|$ upwards, which is expected to perform well when $|\Cl|$ is a good estimate for the optimal $n$, as is the case
in \cite{abelMEMINSATbasedExact2015}.

\subsection{Complexity of the Tail Minimization Problem}

As evidenced in \cref{obs:npmin}, deciding
whether the tail in a cascade composition is $n$-replaceable is an NP problem. For completeness sake we show that this result is tight, meaning that the problem is NP-hard as well. To the best of our knowledge,
this complexity result has not been shown elsewhere.

\begin{theorem}
Deciding whether the tail in a 
cascade composition is $n$-replaceable is an NP-hard problem.
\end{theorem}
\begin{proof}
Let $N$ be an IS machine and let $n\in \N$.
We build machines $H$ and $T$
in polynomial time
satisfying that $N$ is $n$-implementable if and only if $T$ is $n$-replaceable in
$T\circ H$. As deciding whether $N$
is $n$-implementable is an NP-complete problem 
\cite{Pfleeger_1973},
this reduction proves the theorem. 
Informally, the aim is to build $H$
whose output language coincides with
$\Omega_N$, and use an arbitrary implementation of $N$ as $T$. This idea
does not quite work because it requires $\Omega_N$ to have no maximal words, which not may be the case, but the problem can be fixed adding some extra
output symbols and transitions. 
Let $Y$ and $Z$ be the input and output alphabets of $N$ respectively. 
Let $\hat{Y}\coloneqq Y\cup \{\bot\}, \, \hat{Z}\coloneqq Z\cup 
\{\bot \}$, where $\bot$ is a fresh symbol. 
We begin by building a machine $H$ from $Y$
to $\hat{Y}$. We set $S_H\coloneqq S_N\cup \{*\}$,
where $*$ is a fresh state, and $r_H\coloneqq
r_N$. Given $(s,y)\in S_N\times Y$, 
we define $\delta_H(s,y)\coloneqq \delta_N(s,y)$
and $\lambda_H(s,y)\coloneqq \lambda_N(s,y)$
if $(s,y)\in D_N$, and
$\delta_H(s,y)\coloneqq *$
and $\lambda_H(s,y)\coloneqq \bot$ otherwise. 
We also define 
$\delta_H(*,y)\coloneqq *$
and $\lambda_H(*,y)\coloneqq \bot$ for all $y$.
It is direct to see that $\Out(H)=\Omega_N \{ \bot \}^*$. Now we build 
another machine $T$ from $\hat{Y}$
to $\hat{Z}$. Let $N^\prime$
be an arbitrary implementation of $N$, built in polynomial time by simply adding the missing transitions. To construct $T$ we add self loops $\delta_T(s,\bot)=s$, $\lambda_T(s,\bot)=\bot$
to $N^\prime$ at each state $s\in S_{N^\prime}$.\par
Now we show that $T$ is $n$-replaceable
in $T\circ H$ if and only
if $N$ is $n$-implementable. As exposed in
\cref{lemm:main}, a machine $T^\prime$ is a valid replacement for $T$
if and only if for all $\overline{y}\in \Out(H)$
$\lambda_{T^\prime}(\overline{y})
= \lambda_T(\overline{y})$. Moreover, any word $\overline{y}\in \Out(H)$
is of the form $\overline{u}(\bot)^k$, 
where $\overline{u}\in \Omega_N$.
By construction $\lambda_T(\overline{u}(\bot)^k)=
\lambda_N(\overline{u})(\bot)^k$.
This implies the following: 
(1) Let $T^\prime$
be a replacement for $T$. Then removing all transitions on input $\bot$ from $T^\prime$
yields an implementation of $N$ of the same size. (2) Let $T^\prime$ be an implementation of
$N$. Then adding self-loops   $\delta_T(s,\bot)=s$, $\lambda_T(s,\bot)=\bot$
to all states of $T^\prime$ yields a replacement
for $T$ in $T\circ H$ with the same number of states. This proves the result.
\end{proof}



\section{Synthesis of the Tail}
\label{sec:tailsynthesis}
In this second part of the paper we study 
the Tail Synthesis Problem,
and its associated Feasibility Problem. Our main result is the fact
that while the Feasibility Problem has polynomial time complexity, there are instances of the Synthesis Problem
where minimal solutions have exponential size, and hence the problem itself has exponential complexity. The proof
of this result relies on a construction that,
given a feasible instance of the Synthesis Problem, produces an OM that represents all its solutions. \par

The Tail Synthesis problem is a particular case of an ``unknown component equation'',
\cite{villaUnknownComponentProblem2012}.
In the general problem, a component $H$
and a system model $M$ are known, and
the goal is to find $T$ that connected to 
$H$ in a given way yields $M$.
The general approach 
for solving these equations is given in
\cite{watanabeMaximumSetPermissible1993a}.
This method is based on the computation
of the `$E$-machine", which is a DFA $E$ 
over input/output satisfying the
following condition. A Mealy machine $T$ is a
solution for the component equation if and only if all its traces $\langle \overline{y}, \lambda_T(\overline{y})\rangle$ are accepted
by $E$. Deciding whether such $T$
exists amounts to determining the winner of
a \textbf{safety game} on 
$A$, and synthesising $T$ is equivalent to giving 
a winning-strategy for the Protagonist
player in this game. Both these tasks can
be carried out in $O(|A|)$ time \cite{bloem2018graph,bloem2014sat, morgenstern2013solving}. However, building the E-machine
involves the determinization of a product construction, resulting in
$|A|=O(2^{|H|}|M|)$. This procedure gives an $O(2^{|H|}|M|)$ upper complexity bound both for checking whether an equation is feasible and for synthesising
the missing component.
The case where the unknown component
is the head of a cascade composition deserves
special attention, as in that instance
the E-machine can be obtained with no 
determinization \cite{benedettoModelMatchingFinitestate2001},
making the feasibility check and the synthesis task possible within polynomial time in this scenario.

\subsection{Feasibility of the Synthesis Problem}

In this section we characterize the feasible instances of the Tail Synthesis Problem, and show that the
feasibility check can be carried out in polynomial time.
\par 
Let $H$ and $M$ be Mealy machines from $X$ to $Y$ and from $X$ to $Z$ respectively.
A solution $T$ to the synthesis problem (that is, $T\circ H\equiv M$) must satisfy
$\lambda_T(\lambda_H(\overline{x}))=
\lambda_M(\overline{x})$ for all $\overline{x}\in X^*$. In particular, if such solution $T$ exists, then there cannot be
two words $\overline{x}, \overline{x}^\prime$ with $\lambda_H(\overline{x})=\lambda_H(\overline{x}^\prime)$ but 
$\lambda_M(\overline{x})\neq\lambda_M(\overline{x}^\prime)$. We argue that the converse holds as well.

\begin{proposition} \label{thm:feasible}
There exists a Mealy machine $T$
with $T\circ H \equiv M$ if and only if
for any two words $\overline{x}, \overline{x}^\prime\in X^*$ with
$\lambda_H(\overline{x})=\lambda_H(\overline{x}^\prime)$ it holds 
$\lambda_M(\overline{x})=\lambda_M(\overline{x}^\prime)$ as well. 
\end{proposition}
\begin{proof}
See \cref{ann:feasible_proof}.
\end{proof}

As a consequence of this result, deciding the feasibility of
the Tail Synthesis Problem given by  $H$ and $M$ is equivalent to checking whether
$\lambda_H(\overline{x})=\lambda_H(\overline{x}^\prime)$, but
$\lambda_T(\overline{x})\neq\lambda_T(\overline{x}^\prime)$
for some $\overline{x}, \overline{x}^\prime\in X^*$.
The existence of such words $\overline{x},
\overline{x}^\prime$ 
can be easily computed in $O(|H|^2|M|^2)$
time via a fixed point procedure on the synchronous product of
$H$ and $M$.




\subsection{Representing all Solutions}
\label{def:env_aut}

We give the construction of an OM with
size $O(|H||M|)$
encoding all solutions for a 
feasible instance of the
Tail Synthesis Problem. We note that this
is an exponentially more succinct 
representation of the solutions than the E-machine from \cite{watanabeMaximumSetPermissible1993a}.
In fact, this construction 
is equivalent to the NDE-machine
introduced in \cite{harrisSynthesisFiniteState1998,watanabeLogicOptimizationInteracting1994} but with
universal acceptance conditions, rather than
non-deterministic ones. \par



We define an OM $N$ from $Y$ to $Z$
as follows. Let
$S_N \subseteq S_H\times S_M$ 
be the set of pairs $(s_H, s_M)$
that are reachable in the synchronous product $H\times M$.
That is, those satisfying
$\delta_H(\overline{x})=s_H$ and
$\delta_M(\overline{x})=s_M$
for some $\overline{x}\in X^*$.
Let $r_N\coloneqq (r_H,r_M)$. 
We define the transition and output functions $\Delta_N$, $\lambda_N$ for $N$.
Fix a transition $(s_H,s_M)\in S_N, y\in Y$. Let $V$
be the set of $x\in X$ 
satisfying $\lambda_H(s_H,x)=y$.
We take two cases into consideration. If the set $V$ is empty, then we set
the transition as undefined $((s_H,s_M),y)\notin D_N$.
Otherwise, if $V\neq \emptyset$,
we mark the transition as defined
$((s_H,s_M),y)\in D_N$.
In this case $\lambda_M(s_M,x)$ takes a unique value $z$ for all $x\in V$. Indeed, the opposite would yield two sequences 
$\overline{x}, \overline{x}^\prime\in X^*$
with $\lambda_H(\overline{x})=
\lambda_H(\overline{x}^\prime)$ 
and $\lambda_H(\overline{x})
\neq \lambda_M(\overline{x}^\prime)$, making the Tail Synthesis Problem given by $H$ and $M$
infeasible. Hence, we can define
$\lambda_N((s_H,s_M),y)$ as $z$, and $\Delta_N((s_H,s_M),y)$ as
$\{ \,
(s_H^\prime, s_M^\prime) \, \mid
\, \exists x\in V \text{ s.t. }
s_H^\prime=\delta_H(s_H,x),\,
s_M^\prime=\delta_M(s_M,x)
\, \}$.

\begin{proposition} 
Let $N$ be the OM defined above.
A machine $T$ satisfies 
$T\circ H \equiv M$ if and only if
$T$ implements $N$.
\end{proposition}
\begin{proof}
The machine $T$ satisfies $T\circ
H\equiv M$ if and only if
for all $\overline{y}\in \Out(H)$
and $\overline{x}\in X^*$
with $\lambda_H(\overline{x})=\overline{y}$, it holds 
$\lambda_T(\overline{y})=\lambda_M(\overline{x})$. Note that by construction $\Omega_N=\Out(H)$. 
Moreover, there is a run of $N$ over
$\overline{y}$ whose output is $\overline{z}$ if and only if there is some $\overline{x}\in X^*$ 
satisfying $\lambda_H(\overline{x})=\overline{y}$
and $\lambda_M(\overline{x})=\overline{z}$. This shows the result. 
\end{proof}

\subsection{ Lower Bounds for Synthesising the Tail }
In this section we show
that there are instances of the synthesis problem
(\cref{prob:equation}) where all solutions have at least exponential size. 

\begin{theorem} \label{thm:lower_bound}
There exist finite alphabets $X,Y,Z$ and 
families of 
increasingly large
Mealy machines $\{M_n\}_n$, 
from $X$ to $Y$, and $\{H_n\}_n$,
from $X$ to $Z$,
for which the size of
any $T_n$ satisfying $T_n\circ H_n\equiv M_n$
is bounded from below by an exponential function of $|M_n||H_n|$. 
\end{theorem}

The proof of this result has two
parts.
First, we show an infinite
family of OMs
for which all implementations have exponential size (\cref{thm:big_spec}). 
Afterwards we prove that any OM $N$ can be ``split'' into Mealy
machines $M$ and $H$ with the same
number of states (plus one) as $N$
for which any $T$
with $T\circ H\equiv M$ provides an implementation of $N$ (\cref{thm:split}).
These two results together prove \cref{thm:lower_bound}. Some care has to be
employed in order to obtain fixed alphabets
in \cref{thm:lower_bound}. The alphabets for
$H$ and $T$ in the splitting construction
of \cref{thm:split} depend on the alphabets
of $N$ and its \textbf{degree}, defined as $d(N)\coloneqq \max_{(s,y)\in D_N} |\Delta_N(s,y)|$. Hence, it is necessary
to ensure in \cref{thm:big_spec}
that the resulting OMs have bounded degree.
 
\begin{lemma}
\label{thm:big_spec}
There are finite alphabets
$Y,Z$ and increasingly large
OMs $\{ N_n \}_n$ from $Y$
to $Z$ for which the size
of a machine $T_n$ implementing
$N_n$ is bounded from below 
by an exponential function of $|N_n|$.
Moreover, it is possible to build
$N_n$ in such a way that
$\max_n d(N_n) = 2$ 
\end{lemma}
\begin{proof}

See \cref{ann:big_aut}.
\end{proof}

\begin{lemma}
\label{thm:split}
Let $N$ be a consistent
OM from $Y$ to $Z$, and let $k\coloneqq d(N)$. 
Let $X\coloneqq Y\times [k]$, 
$\hat{Y}=Y\cup \{\bot \}$, and
$\hat{Z}=Z\cup \{ \bot \}$, where $\bot$ is a fresh symbol.
Then there exist
Mealy machines $H$, from
$X$ to $\hat{Y}$,
and $M$, from $X$ to $\hat{Z}$,
such that
(1) $|S_H|,|S_M|= |S_N|+1$, 
(2) there are machines
$T$ with $T\circ H \equiv M$,
and
(2) any of those machines $T$
satisfies $\lambda_T(\overline{y})=\lambda_N(\overline{y})$ for all
$\overline{y}\in \Omega_N$.
\end{lemma}
\begin{proof}
We give a explicit construction for $H$ and $M$.
We define the
edge set of $N$, 
$G_N\subset S_N\times Y \times S_N$, as the set consisting of
the triples $(s,y,s^\prime)$
where $(s,y)\in D_N$ and
$s^\prime\in \Delta_N(s,y)$.
By the definition of $k$, 
we can build a map $L:G_N\rightarrow [k]$
satisfying 
$L(s,y,t_1)\neq L(s,y,t_2)$, 
for any $(s,y)\in D_T$ and 
any two different states $t_1, t_2\in \Delta_N(s,y)$. The map $L$ assigns a number 
$0\leq i < k$ to each edge
$(s,y,s^\prime)\in G_N$, giving different labels to each edge corresponding
to a pair $(s,y)\in D_N$.
We define $H$ and $M$ at the same time. Set $S_H,S_M\coloneqq S_N\cup \{*\}$, where $*$ is
a fresh state, and $r_H,r_M\coloneqq r_N$. 
Let $s\in S_N\cup \{*\}$ and let
$x\coloneqq (y,i)\in X=Y\times [k]$. We take into account three cases: 
(1) If $s=*$, then $\delta_H(*,x),\delta_M(*,x)=*$,
and $\lambda_H(*,x),\lambda_M(*,x)=\bot$. (2)
If $s\in S_N$ and there exists some $t\in S_N$ with $L(s,y,t)=i$, then 
$\delta_H(s,x),\delta_M(s,x)=t$,
$\lambda_H(s,x)=y$, and $\lambda_M(s,x)=
\lambda_N(s,y)$. Finally, 
(3) if $s\in S_N$ and there is no
$t$ with $L(s,y,t)=i$, then
$\delta_H(s,x),\delta_M(s,x)=*$,
$\lambda_H(s,x),\lambda_M(s,x)=\bot$.
We claim that $H$ and $M$ built this way
satisfy the theorem's statement
(see \cref{app:split}).
\end{proof}


\section{Other Compositions}

The results and techniques exposed so far pertain to the optimization and synthesis of the tail component in cascade compositions. However, it is possible to 
study the analogous problems in a more general setting through polynomial reductions. 
Consider a system consisting of two interconnected components $H$ and $T$, where
all of $T$'s output signals can be externally observed. Without loss of generality we can assume that (1) $T$'s input signals coincide with $H$'s output signals, (2) the system's output signals coincide with $T$'s output
signals, and (3) $H$'s input
signals are the system input signals plus $T$'s output ones. This situation is depicted in \cref{fig:general_cascade}. We claim that
the minimization and synthesis of $T$ in this context can be polynomially
reduced to those of the tail in a cascade composition. The reductions are similar to the ones in \cite{Wang_Brayton_1993}, and can be found in \cref{ann:reductions}.

\section{Experimental Evaluation}
Since the Tail Minimization Problem is in NP, it allows for a straightforward reduction to SAT, which already improves over the doubly-exponential complexity of the K-N algorithm. We show in
\cref{fig:naive} preliminary experimental results 
comparing this solution 
with our proposed one. The baseline method uses a ``naive'' CNF encoding to decide whether the tail component is $n$-replaceable. Like in our proposed approach, this is done
for increasing values of $n$ until a satisfiable CNF is obtained. The encoding can be seen as analogous to ours but without the information about the incompatibility graph and the partial solution. 
The experiments were run on cascade compositions consisting
of independently randomly generated machines $M$ and $T$
with $n=|S_M|=|S_T|$ and input/output alphabets of size 4.
Mean CPU times of all runs for each $n$ are plotted for both the baseline algorithm and ours. 
The baseline implementation is not able to complete any instance with $n \geq 12$
with a timeout of $10$ minutes, while our algorithm solved all
instances in under a minute. We conclude that our approach has a clear benefit over a straightforward approach for this class of random instances. 
Both the implementation of our algorithm and the baseline
algorithm use CryptoMiniSat \cite{DBLP:conf/sat/SoosNC09}. 
All experiments were run on a Intel Core i5-6200U (2.30GHz) machine.

\begin{figure}[tbh]
\centering
\begin{subfigure}{0.45\textwidth}
\includegraphics[width=\textwidth]{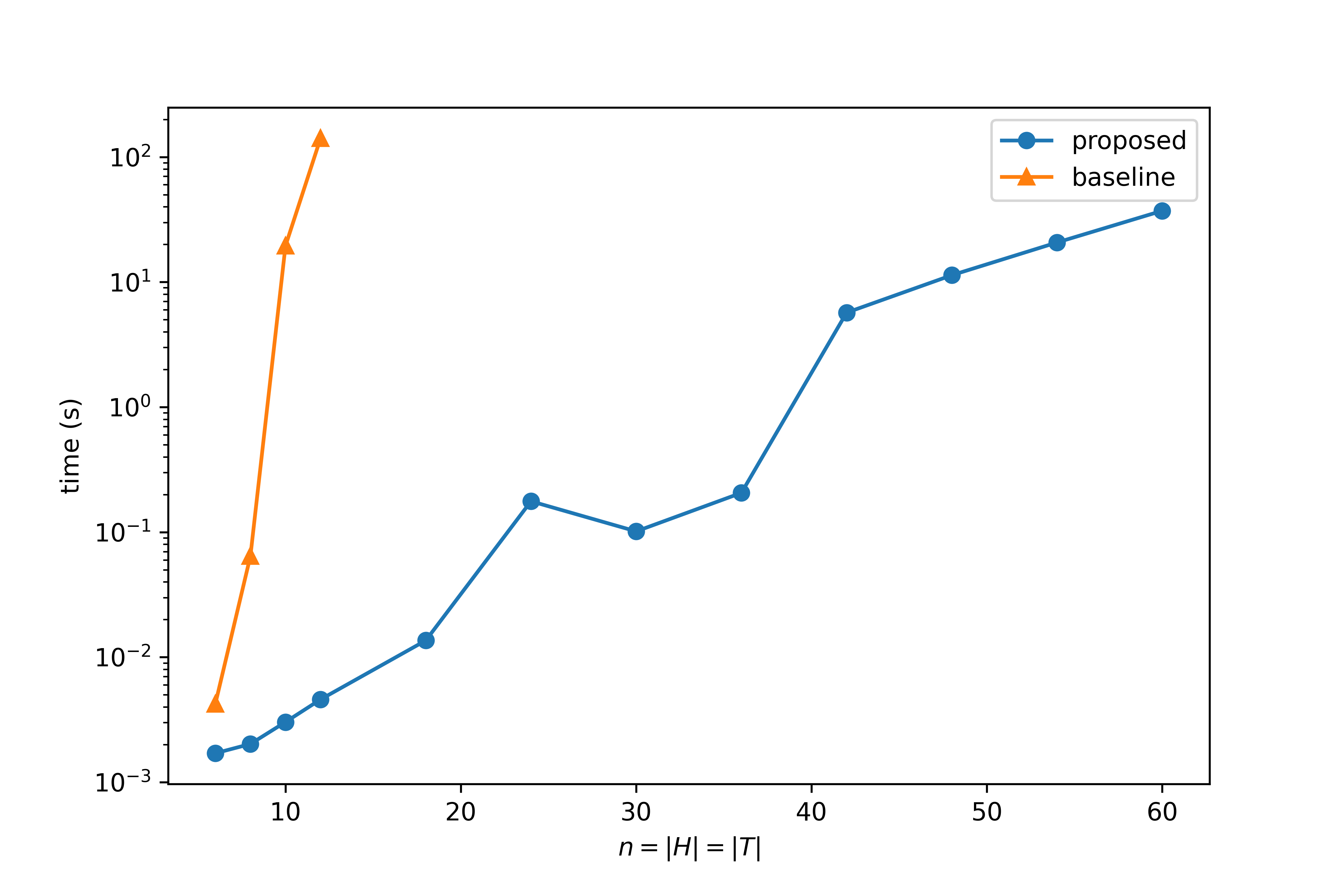}
 \caption{A comparison between our proposed algorithm and a the baseline one.}
\label{fig:naive}
\end{subfigure}
    \hfill
    \begin{subfigure}{0.45\textwidth}
    \includegraphics[width=\textwidth]{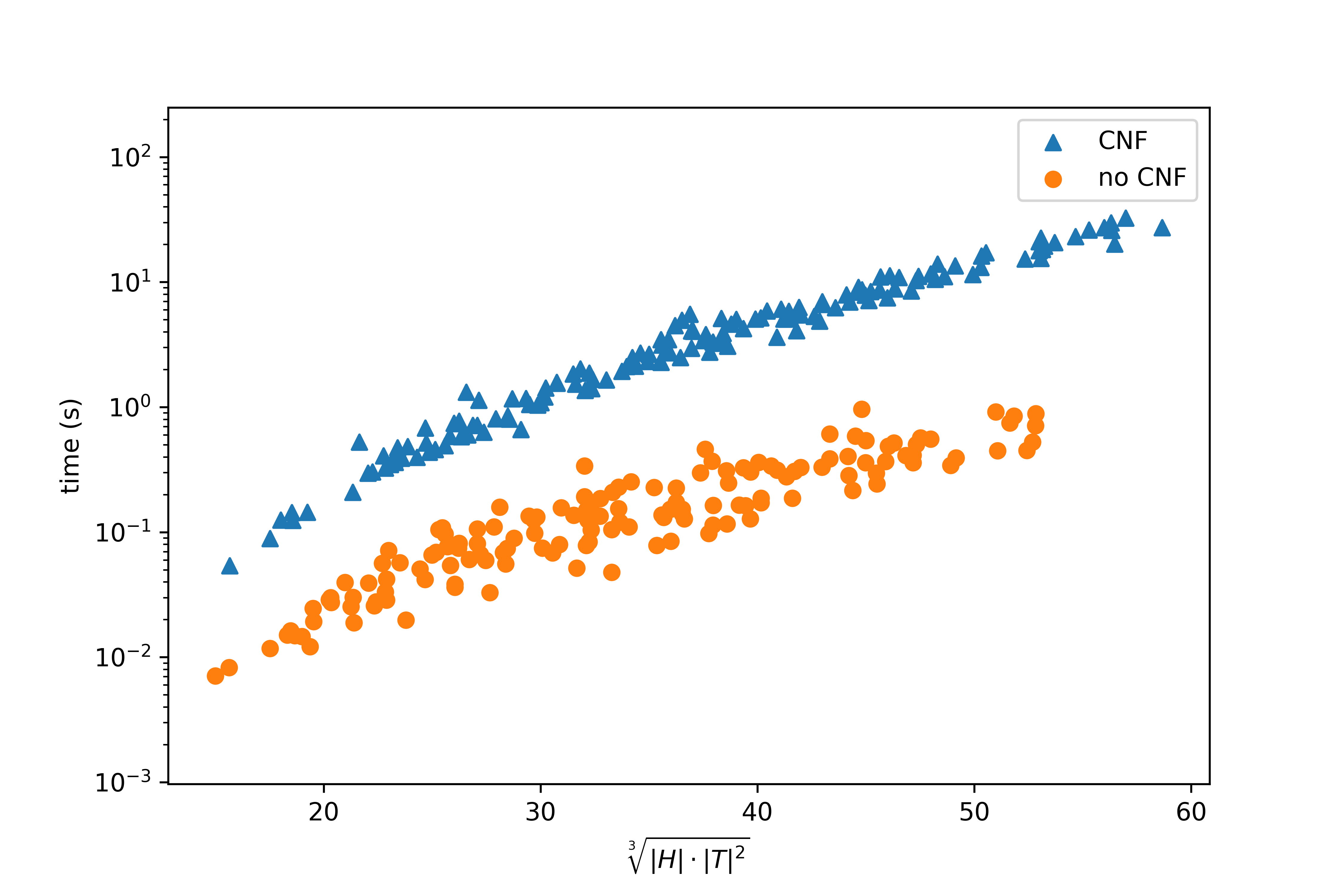}
     \caption{Random experiments run with our proposed minimization method.
     }
    \label{fig:bimodal}
    \end{subfigure}
\end{figure}
 
One reason for the good performance is that our algorithm skips the CNF
encoding entirely 
in about half the instances by 
using the size of the partial solution, 
as discussed in \cref{sec:mini_end}. 
To illustrate this we ran 200
additional experiments where
$T$ and $H$ were generated independently with random sizes between 12 and 60 states. The
running times are shown in \cref{fig:bimodal}, where
orange points correspond to
instances where
it was possible to avoid any CNF encoding, and the blue ones correspond to the rest. It can
be observed that orange and blue points form two separate clouds of points. It is apparent that the algorithm is significantly more efficient when it does not call the SAT solver, but the approach uses under a minute regardless of whether the SAT solver is employed.

\newpage

\bibliography{main}

\newpage

\appendix

\section{Horn CNF Encoding for the Incompatibility Relation}
\label{ann:graph}
Let $M$ be an OM from
$Y$ to $Z$.
In this section we describe a way of computing
the compatibility relation over $M$ defined
in \cref{sec:covering}
by finding the minimal
satisfying assignment of a Horn formula.
Let $m\coloneqq 
|S_M|$. We identify
states in $S_M$ with numbers $j\in [m]$. We describe our
Horn encoding as follows. For each $i,j \in [m]$
with $i\leq j$ we introduce a variable $X_{i,j}$ meaning $i\nsim j$.
We give now the clauses of the formula. For each pair $i\leq j$ we consider two scenarios:
(1) If $i$ and $j$
are incompatible
at depth one, we add the clause $X_{i,j}$. (2) Otherwise, for all
$y\in Y,
i^\prime \in \Delta_M(i,y),
j^\prime \in \Delta_M(j,y)
$ we add clauses
$X_{i^\prime,j^\prime}
\implies X_{i,j}
$,
where the values of $i^\prime$ and $j^\prime$ are swapped
if necessary to ensure $i^\prime \leq j^\prime$.\par
It follows from 
\cref{lem:compat_2}
that two states
$i\leq j$ are compatible if and only if $X_{i,j}$ is true in the minimal satisfying assignment for this encoding. It is well-known that assignment can be obtained in time linear in the size of the formula \cite{dowling1984linear}. Moreover, the encoding has size $O(|M|^2)$. Hence,
the compatibility relation over $M$ can be computed in $O(|M|^2)$ time.

\section{Correctness of the Covering Reduction}
\label{ann:covering}
\begin{proof}[Proof that the $N$ constructed in
\cref{thm:cover_reduction} implements $M$:]
Let $\overline{y}\in \Omega_M$. Consider
a run $s_0\coloneqq r_M, y_0, z_0, s_1, \dots, s_n,
y_n, z_n, s_{n+1}$ of $M$ on $\overline{y}$, 
and let 
$C_0\coloneqq r_N, y_0, z^\prime_0, C_1, \dots, C_n,
y_n,$ $z^\prime_n, C_{n+1}$ be the run of $N$ on $\overline{y}$. By construction, 
$s_i\in C_i$ for all $0\leq i \leq n+1$. Moreover, as $\lambda_M(s_i,y_i)=z_i$ it also holds
$\lambda_N(C_i,y_i)=z_i$. Thus, $\lambda_M(\overline{y})=\lambda_N(\overline{y})$
and $N$ implements $M$.
\end{proof}

\section{CNF Encoding for OM Minimization}
\label{ann:minim_cnf}

We describe how to reduce the problem of finding a closed cover of $n$ over
$M$ to SAT, following the CNF encoding given in \cite{abelMEMINSATbasedExact2015}.
We identify the set $S_M$ with the set of integers
$[m]$, where $m\coloneqq |S_M|$.
Finding a closed cover of $n$ compatibles
is equivalent to finding two maps
$C:[n]\rightarrow 2^{[m]}$ and $\mathrm{Succ}:[n]\times Y \rightarrow 2^{[n]}$ that satisfy:
(1) for each $i\in [n]$ the set $C(i)\subseteq [m]$ is a compatible,
(2) $r_M \in C(i)$ for some $i$, 
(3) $\Delta_M(C(i),y)\subseteq C(j)$
for all $j\in \mathrm{Succ}(i,y)$, and
(4) $\mathrm{Succ}(i,y)\neq \emptyset$
for each $i\in [n],y\in Y$. 
The propositional variables of the CNF encoding are the following:
A variable $L_{s,i}$ for each $s\in [m], i\in [n]$,
encoding that $s\in C(i)$. 
A variable $N_{i,j,y}$ for each pair $i,j\in [n], y\in Y$, encoding that $j\in Succ(i,y)$. 
The clauses of our encoding are as follows:
A clause 
$\neg L_{s_1,i} \vee \neg L_{s_2,i}$,
for each $i\in [n], s_1,s_2\in [m]$ with 
$s_1\leq s_2$ and $s_1\nsim s_2$. This encodes
condition (1). 
A clause
$\vee_{i\in [n]} L_{r,i}$, which encodes condition (2). 
We have a clause
$\vee_{j\in [n]} N_{i,j,y}$ for each $i\in [n], y\in Y$, encoding condition (3).
Finally, there is a clause
$(N_{i,j,y} \wedge L_{s,i}) \implies \vee L_{s^\prime,j}$
for each $s\in [m]$, $y\in Y$ with $(s,y)\in D_M, s^\prime\in \Delta_M(s,y)$, $i,j\in [n]$. These clauses encode condition (4).  The CNF obtained so far already encodes the desired  problem. However, we also 
use a partial solution $\Cl \subseteq 
S_M$ for adding symmetry breaking
predicates. If $\Cl=\{s_1,\dots,s_l\}$
is a set of $l\leq n$ pair-wise incompatible states, we can add the clauses $L_{s_i,i}$ for each $0\leq i \leq l$.

\section{Characterization of Feasible Synthesis Instances}
\label{ann:feasible_proof}
\begin{proof}[Proof of \cref{thm:feasible}]
The fact that $T$'s existence implies the second part of the statement is straightforward,
as discussed above. We prove the other implication. Suppose 
that any two words $\overline{x}, \overline{x}^\prime$ with
$\lambda_H(\overline{x})=\lambda_H(\overline{x}^\prime)$ also satisfy 
$\lambda_M(\overline{x})=\lambda_M(\overline{x}^\prime)$. This defines
a map $F:\Out(H)\rightarrow Z^*$
by setting $F(\overline{y})=\lambda_M(\overline{x})$ if $\overline{y}=\lambda_H(\overline{x})$. Let us define the language $\Lcal_F\subseteq (Y\times Z)^*$ 
as the one consisting of the words
$\langle \overline{y},
F(\overline{y})\rangle$ for all
$\overline{y}\in \Out(H)$. Clearly this
language is regular and prefix-closed. Hence, there is some Mealy machine
$T$ form $Y$ to $Z$ satisfying 
$F(\overline{y})=\lambda_T(\overline{y})$ for all $\overline{y}\in \Out(H)$.
By definition of $F$, we have
$\lambda_T(\lambda_H(\overline{x}))=
\lambda_M(\overline{x})$ for all $
\overline{x}\in X^*$, and $T\circ H
\equiv M$.
This proves the result. 
\end{proof}

\section{Observation Machines with no Small Implementations}
\label{ann:big_aut}
Fix $n>0$. \Cref{fig:exp} shows the construction of an OM $M$
(bottom right) from $Y\coloneqq\{a,b\}$ to $Z\coloneqq\{a,b,\top\}$
for which all implementations
have at least $2^n$ states. Moreover, $d(M)=2$.
Let $y_0 y_1 \dots y_{2n-1}$ be a word in $Y^*$. 
$M$ behaves the following way:
(1) It outputs $\top$ in response to the first $n$ inputs $y_0,\dots, y_{n-1}$.
(2) Starting from $y_{n}$,
$M$ responds with $\top$ until
$b$ is received as an input.
(3) If $y_{n+i}$ is the first input equal to $b$ 
since $y_{n}$, then $M$ outputs $y_i$ in response. Intuitively, an implementation of $M$ has to have at least $2^n$ states because it has to store the first $n$ inputs in order to carry out (3).

\label{sec:exp_obs}
\begin{figure}[tbh]
    \centering
    \begin{subfigure}{0.45\textwidth}
    \includegraphics[width=\textwidth]{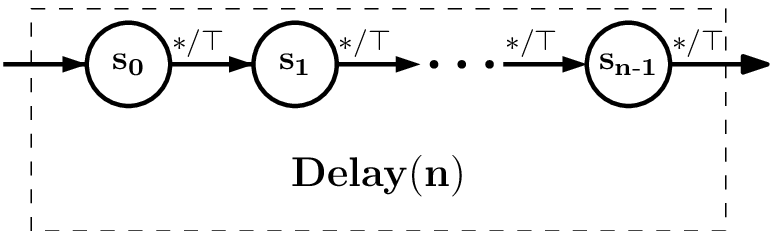}
    \vspace{0.3cm}
    \end{subfigure}
    \hfill
     \begin{subfigure}{0.50\textwidth}
    \includegraphics[width=\textwidth]{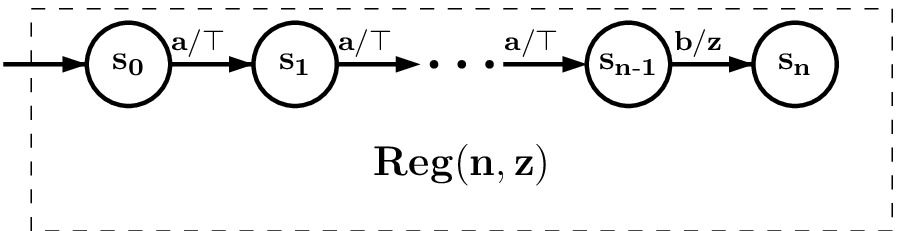}
    \vspace{0.3cm}
    \end{subfigure}
    \\
    \vspace{0.3cm}
    
    \begin{subfigure}{0.3\textwidth}
    \includegraphics[width=\textwidth]{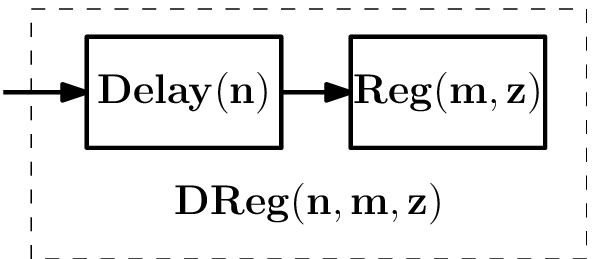}
    \vspace{0.3cm}
    \end{subfigure} 
    \hfill
    \begin{subfigure}{0.65\textwidth}
    \includegraphics[width=\textwidth]{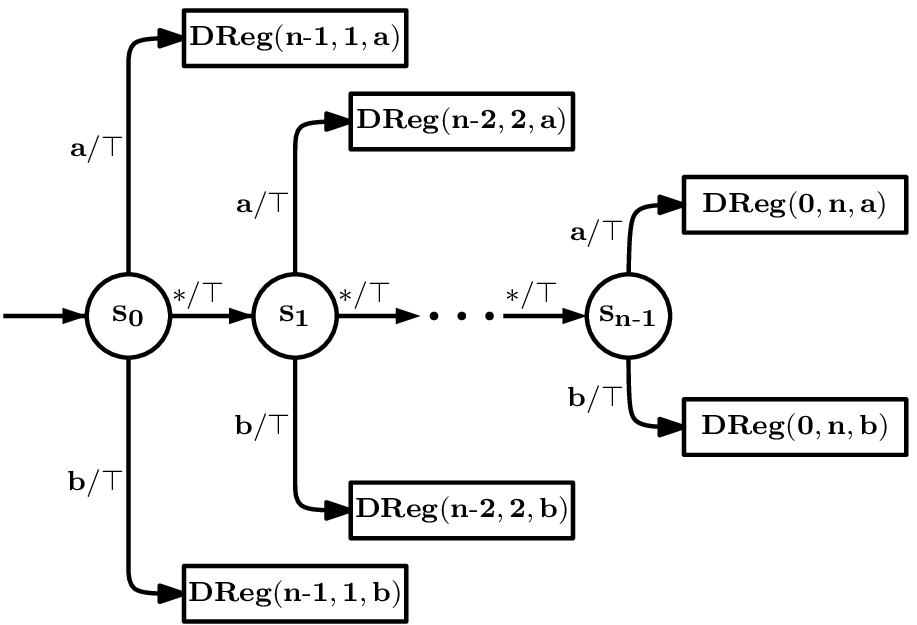}
    \vspace{0.3cm}
    \end{subfigure}
    \caption{The construction of an OM with no small implementations (bottom right). The symbol $*$ stands for both $a,b$.}
    \label{fig:exp}
\end{figure}

\section{Correctness of the splitting construction}
\label{app:split}
\begin{proof}[Correctness of the constructions in
\cref{thm:split}]
Here we show that the given constructions for $H$ and $M$
indeed fulfill the result. We have to show that (1)
some $T$ satisfies $T\circ H \equiv M$ and (2) any such 
$T$ also satisfies $\lambda_N(\overline{y})=
\lambda_T(\overline{y})$ for all $\overline{y}\in
 \Omega_N$.
Remember that characters $x\in X$ are 
(input,label) pairs $(y,l)\in Y\times [k]$.
We say that a word $\overline{x}\coloneqq
(y_0,l_0)\dots (y_n,l_n)
\in X^*$ labels a run of $N$
if $N$ has a run 
$r_N= s_0, y_0, z_0, s_1,\dots, s_n, y_n, z_n, s_{n+1}$
where $L(s_i,y_i,s_{i+i})=l_i$ for all $0\leq i \leq n$.
Given such $\overline{x}$ by construction
$s_0, (y_0,l_0), y_0, s_1, \dots, s_n,
(y_n,l_n), y_n, s_{n+1}$ and
$s_0, (y_0,l_0), z_0, s_1, \dots, s_n,
(y_n,l_n), z_n, s_{n+1}$ are the runs
of $H$ and $M$ on $\overline{x}$, respectively. 
Let $\overline{x}\in X^*$ be arbitrary. We can write
$\overline{x}=\overline{u}\overline{v}$, where
$\overline{u}=\langle \overline{y},
\overline{l} \rangle$ is the largest prefix 
of $\overline{x}$ which labels a run of $N$.
The previous observation yields 
$\lambda_H(\overline{u})=\overline{y}$,
$\lambda_M(\overline{u})=\lambda_N(\overline{y})$. 
Moreover, by construction
$\lambda_H(s,\overline{v}),\lambda(s,\overline{v})$
are both sequences of only $\bot$ symbols,
where $s=\delta_H(\overline{u})=\delta_M(\overline{u})$.
This way we have shown that the language 
$E(H,M)=\{ \, \langle \lambda_H(\overline{x}),
\lambda_M(\overline{x})\rangle \, \mid \,
\overline{x}\in X^* \, \}$ equals
$\{ \,
\langle \overline{y},
\lambda_N(\overline{y})\rangle \,
\mid \, \overline{y}\in \Omega_N \,
\} \{ \, \bot \, \}^*
$.
This identity proves both (1) and (2). Indeed, if $\lambda_H(\overline{x})=\lambda_H(\overline{x}^\prime)$,
then necessarily $\lambda_M(\overline{x})=\lambda_M(
\overline{x}^\prime)$. By \cref{thm:feasible} this implies (1).
Also, if $\lambda_H(\overline{x})=\overline{y}$ and
$\overline{y}
\in \Omega_N$,
then $\lambda_M(\overline{x})=\lambda_N(\overline{y})$, which shows
(2). 
\end{proof}

\section{Reductions to Cascade Compositions}
\label{ann:reductions}
We put $T\odot H$ for the composition between two Mealy machines $H$ and $T$ shown in
\cref{fig:general_cascade}.
This composition is not well-defined in general \cite[Section 6.2]{harrisSynthesisFiniteState1998}, but a sufficient requirement is that $H$ is a Moore machine, for example.
Analogously to \cref{prob:minim},
we consider the problem of finding a minimal replacement for $T$ in $T\odot H$ when both $T$
and $H$ are given. In \cref{fig:red1} it is shown how to build a machine $H^\prime$
such that finding a minimal
replacement for $T$ in $T\odot H$ is equivalent to
finding a minimal replacement for $T$ in the cascade composition
$T \circ H^\prime$.
Similarly, when $M$ and $H$
are given, we can study the problem of finding $T$
with $T\odot H \equiv M$. In 
\cref{fig:red2} it is shown how to build $H^\prime$ in a way that $T$ satisfying $T\odot H \equiv M$ is equivalent to 
$T\circ H^\prime \equiv M$.

\begin{figure}[tbh]
    \centering
    \begin{subfigure}{0.40\textwidth}
    \includegraphics[width=\textwidth]{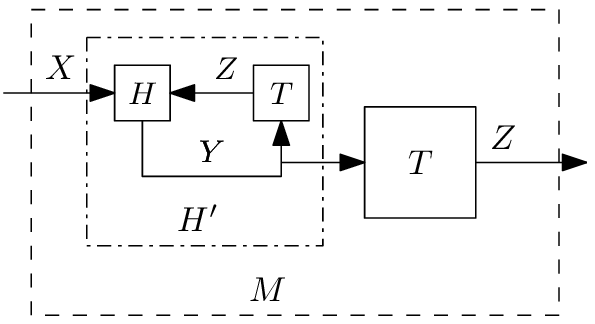}
    \caption{}
    \label{fig:red1}
    \end{subfigure} 
    \hfill
    \begin{subfigure}{0.40\textwidth}
    \includegraphics[width=\textwidth]{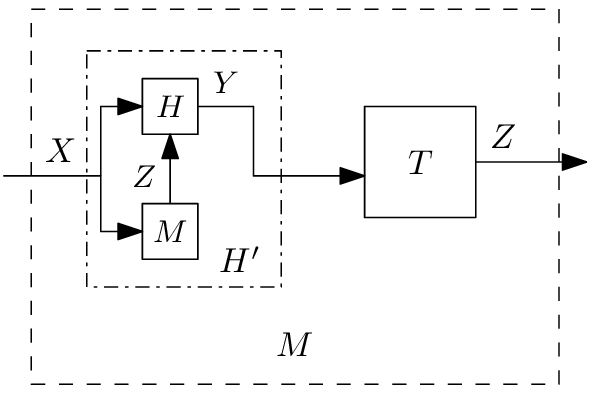}   
    \caption{}
    \label{fig:red2}
    \end{subfigure}
    \caption{Polynomial transformations to cascade compositions.}
\end{figure}

\end{document}